\documentclass[aps,pra,10pt,twocolumn]{revtex4-1}

\usepackage[utf8]{inputenc}
\usepackage[OT4]{polski}
\usepackage{amsmath,amssymb}
\usepackage[english]{babel}
\usepackage{graphicx}
\usepackage{dcolumn}
\usepackage{mathtools}
\usepackage[caption=false]{subfig}
\usepackage[squaren, cdot]{SIunits}
\usepackage{physics}
\usepackage{dsfont}
\usepackage{slashed}
\usepackage{amsthm}
\usepackage{mathrsfs}
\usepackage{xcolor}
\usepackage{hyperref}
\hypersetup{
    colorlinks,
    linkcolor={red!50!black},
    citecolor={blue!50!black},
    urlcolor={blue!80!black}
}
\usepackage{float}
\usepackage{bm}
\usepackage{cleveref}
\usepackage{csquotes}

\Crefname{equation}{Eq.}{Eqs.}
\Crefname{figure}{Fig.}{Figs.}
\Crefname{tabular}{Tab.}{Tabs.}

\renewcommand{\vec}{\bm}
\newcommand{\eq}[1]{Eq.~\eqref{#1}}
\newcommand{\fig}[1]{Fig.~\ref{#1}}
\newcommand{\sect}[1]{Sec.~\ref{#1}}

\newcommand{\ie}{i.\thinspace{}e. }

\begin{document}

\title{Rapid generation of Mott insulators from arrays of noncondensed atoms}

\author{M. R. Sturm}
\email[Corresponding author.\\]{martin.sturm@physik.tu-darmstadt.de}
\author{M. Schlosser}
\author{G. Birkl}
\author{R. Walser}
\affiliation{Institut f\"ur Angewandte Physik, Technische Universit\"at Darmstadt, 64289 Darmstadt, Germany}
\date{\today}

\begin{abstract}
We theoretically analyze a scheme for a fast adiabatic transfer of cold atoms from the atomic limit of isolated traps to a Mott-insulator close to the superfluid phase. This gives access to the Bose-Hubbard physics without the need of a prior Bose-Einstein condensate. The initial state can be prepared by combining the deterministic assembly of atomic arrays with resolved Raman sideband cooling. In the subsequent transfer the trap depth is reduced significantly. We derive conditions for the adiabaticity of this process and calculate optimal adiabatic ramp shapes. Using available experimental parameters, we estimate the impact of heating due to photon scattering and compute the fidelity of the transfer scheme. Finally, we discuss the  particle number scaling behavior of the method for preparing low-entropy states. Our findings demonstrate the feasibility of the proposed scheme with state-of-the-art technology.
\end{abstract}

\maketitle

\section{Introduction}
Deterministic preparation of cold atoms in optical microtrap arrays \cite{Endres2016, Barredo2016, Kim2016, Robens2017a, OhldeMello} combined with Raman-sideband cooling \cite{Kaufman2012, Thompson2013} constitutes a promising source for low-entropy many-body states \cite{Olshanii2002, Weiss2004}. This approach assembles quantum many-body systems atom-by-atom, contrasting the loading schemes used in optical lattice experiments which start from the bulk, \ie Bose-Einstein condensates, or degenerate Fermi gases \cite{Bloch2008, Lewenstein2012}. The assembly of atomic arrays with unit filling and the Raman sideband cooling of the atoms to the respective motional ground state require tight isolated traps, which prohibit inter-site tunneling. Therefore, after the cooling process the trap depth or the trap spacing needs to be reduced significantly, in order to explore the itinerant physics of the Hubbard model. This was demonstrated in double wells for bosonic atom pairs \cite{Kaufman2014} and for fermionic atom pairs \cite{Murmann2015} using the \textquote{spilling technique} \cite{Serwane2011} instead of Raman sideband cooling.

In this article, we investigate the time-dependent transfer of bosonic atoms from an array of isolated traps to a tunnel-coupled array. A detailed analysis shows that this bottom-up approach to a Mott-insulator state is achievable. Reducing the trap depth instead of the trap spacing is preferred, because the latter results in a large overlap of the optical microtraps prohibiting cross-talk-free single-site control \cite{Sturm2017}. Clearly, the time-dependent transfer has to be \textquote{as fast as possible, but as slow as necessary}, to avoid ramp-induced excitations on one hand, and to suppress external heating mechanisms or loss processes, on the other hand. In order to satisfy these conflicting conditions, we derive optimal intensity ramp shapes.

The article is organized as follows: In \sect{sec:adiab}, we formulate an adiabatic variational procedure for optimal time-dependent parameter ramps. In Sec.~\ref{sec:model}, we set up the model for ultracold atoms in optical microtraps, discuss the regimes traversed during the transfer process, and apply the formalism developed in \sect{sec:adiab}. Current experiments with optically trapped atoms are used as benchmarks to obtain realistic system parameters in \sect{sec:example}. Employing these results, we compute an optimal adiabatic ramp in \sect{sec:mott}. For this ramp we estimate the impact of heating due to light scattering and compute the transfer fidelity by solving the time-dependent Schrödinger equation for the one-dimensional Bose-Hubbard model. The particle number scaling behavior of the procedure is discussed in \sect{sec:scale}. Finally, in \sect{sec:conclusions} we summarize our findings and provide an outlook.

\section{Rapid adiabatic parameter ramps}
\label{sec:adiab}

Time-dependent manipulations of atom traps have to be sufficiently slow to avoid  excitations. Therefore, one has to specify the conditions of adiabaticity and define error measures for time-dependent transfer processes. 

Let us consider a quantum system with Hamilton operator $\hat{H}(\gamma)$, which is controlled by a $\ell$-dimensional time-dependent parameter $\gamma(t)$  within a time interval $\tau$. Its instantaneous 
energies $E_i(\gamma)$ and eigenstates $\ket{i(\gamma)}$ are obtained from the stationary Schr\"odinger equation
\begin{align}
\hat{H}(\gamma)\ket{i(\gamma)}&=E_i(\gamma)\ket{i(\gamma)}.
\end{align}
The adiabatic theorem \cite{Born1928,Messiah1961} states that systems prepared initially in the energy eigenstate $\ket{i(\gamma(0))}$ will remain in $\ket{i(\gamma(t))}$, if the rate of change of the parameters $\gamma$ is sufficiently small and the energy levels $E_i(\gamma)$ are well separated. In absence of induced resonant transitions, a sufficient criterion \cite{Marzlin2004, Amin2009} for adiabaticity is given by 
\begin{equation}
\label{eq:adiabaticCriterion}
\max_{0\leq t\leq  \tau} \bigg| \frac{\alpha_{ij} (
\gamma,\dot{\gamma})}{\hbar \omega^2_{ij} (\gamma)} \bigg|^2 \ll 1, \quad 
\forall j \neq i.
\end{equation}
Here, we have introduced the transition frequencies
\begin{equation}
\omega_{ij}(\gamma) = \frac{E_j(\gamma) - E_i(\gamma)}{\hbar}
\end{equation}
and the transition matrix elements
\begin{align}
\label{rotspeed}
\alpha_{ij}(\gamma,\dot{\gamma}) =\bra{j} \partial_{t} \hat{H} \ket{i}
=\sum_{l=1}^\ell \dot{\gamma}_l
\bra{j} \partial_{\gamma_l} \hat{H} 
\ket{i}.
\end{align}
Based on measuring the instantaneous loss out of the state $\ket{i}$ into any other state $\ket{j}$ by
\begin{align}
\label{eq:adiasum}
\mathcal{L}(\gamma,\dot{\gamma}) &= \sum\limits_{j \neq i} 
\bigg|\frac{\alpha_{ij}(\gamma,\dot{\gamma})}{\hbar \omega^2_{ij}(\gamma)}
\bigg|^2,
\end{align}
one can express the  cumulative adiabatic error as
\begin{align}
\label{eq:measureinf}
\mathds{E}_{\infty}[\gamma,\dot{\gamma}]&=
\max_{0\leq t\leq  \tau} \mathcal{L}(\gamma(t),\dot{\gamma}(t)),
\end{align}
within the interval $[0,\tau]$. The smallness of $\mathds{E}_{\infty}$ defines an optimality criterion for adiabaticity (cf. \eq{eq:adiabaticCriterion}) for a time-dependent process $\gamma(t)$, starting from $\gamma (0)$ and reaching $\gamma (\tau)$ within duration $\tau$. 

Alternatively, the time-averaged functional
\begin{align}
\label{eq:measure1}
\mathds{E}_1[\gamma,\dot{\gamma}]&= 
\frac{1}{\tau} \int_0^\tau\mathcal{L}
(\gamma(t),\dot{\gamma}(t)) \; \dd t,
\end{align}
is also a cumulative measure for the non-adiabaticity of the process. Here, the error measures $\mathds{E}_p$ are analogs of $p$-norms $\lVert \vec{x}\rVert_p=\sqrt[p]{\sum_i |x_i|^p}$ for 
finite dimensional vectors $\vec{x}$. Clearly, the definition of \eq{eq:measure1} is more amenable to extremization using variational analysis than the definition of \eq{eq:measureinf}. In Appendix \ref{sec:proof}, we show that for the case of a one-dimensional parameter function, as considered in this manuscript, a parameter curve, which minimizes $\mathds{E}_1$ also minimizes $\mathds{E}_\infty$.

By considering the structure of $\mathcal{L}$ in Eqs. \eqref{rotspeed} and \eqref{eq:adiasum}, one obtains a quadratic form in terms of the velocities $\dot{\gamma}$, 
\begin{align}
\label{eq:adiasumquad}
\mathcal{L}(\gamma,\dot{\gamma}) &=
\sum_{k,l=1}^{\ell}{\frac{1}{2}\dot{\gamma}_k\mathcal{M}_{kl}(\gamma)\dot{
\gamma}_l}
\end{align}
and a symmetric, parameter-dependent 'mass matrix' $\mathcal{M}(\gamma)$ in close analogy to the Lagrangian mechanics. Optimal trajectories $\gamma$ are obtained from the Euler-Lagrange equations
\begin{equation}
\label{euler}
\frac{\text{d}}{\text{d}t} \partial_{\dot{\gamma}_l} \mathcal{L}= 
\partial_{\gamma_l}\mathcal{L}.
\end{equation}
Clearly, we can also introduce a canonical momentum $\pi_i=\partial_{\dot{\gamma}_i}\mathcal{L}= (\mathcal{M} \dot{\gamma})_i$ and obtain a Hamiltonian function
\begin{align}
\label{eq:Hamiltonian_function}
\mathcal{H}(\gamma,\pi)=
\sum_{l=1}^{\ell}{\pi_l\dot
{\gamma}_l-\mathcal{L}}=
\sum_{k,l=1}^{\ell}{\frac{1}{2}\pi_k\mathcal{M}_{kl}^{-1}(\gamma)\pi_l},
\end{align}
via a Legendre transformation. From \eq{eq:Hamiltonian_function} Hamilton's equation of motion can be derived as
\begin{align}
\dot{\gamma}_l&=\partial_{\pi_l}{\mathcal{H}}=(\mathcal{M}^{-1}(\gamma)\pi
)_l, &
\dot{\pi}_l&=-\partial_{\gamma_l}{\mathcal{H}}.
\end{align}
If the system is not subject to any additional external time dependence, then the Hamiltonian function is constant
\begin{align}
\mathcal{H}(\gamma(t),\pi(t))=\mathcal{H}_0.
\end{align}
In the special case of one-dimensional parameter processes $\ell=1$, which is considered in \sect{sec:mott}, this leads to completely integrable dynamics 
\begin{align}
\label{integral}
\int_{\gamma(0)}^{\gamma(t)}{\text{d}\gamma\sqrt{\mathcal{M}(\gamma)}}=
\pm\sqrt{2 \mathcal{H}_0} t
\end{align}
for the optimal adiabatic process $\gamma(t)$.

Our approach is equivalent to the concept of the 'quantum adiabatic brachistochrone' \cite{Rezakhani2009} and is strongly related to constant adiabaticity pulses used in nuclear magnetic resonance \cite{Baum1985}. 

\section{Cold atoms in optical microtraps}
\label{sec:model}

The physics of dilute atomic gases is determined by the interplay of single particle motion in the parameter-dependent external potential $V(\vec{r},\gamma)$ and internal pressure arising from the van-der-Waals interaction \cite{Pethick2002, Proukakis2013, Gardiner2017}. In the $s$-wave limit, the latter can be described by a contact interaction of strength $g=4 \pi \hbar^2 a_s/m$, with the atomic mass $m$ and the scattering length $a_s$. Therefore, the system's Hamilton operator reads
\begin{equation}
\label{eq:many_body_H}
\begin{split}
\hat{H}(\gamma) = &\int  \hat{\Psi}^\dagger(\vec r) H_\text{sp} (\vec{r},\gamma) \hat{\Psi}(\vec r) \; \dd^3 r\\
&+\frac{g}{2} \int  \hat{\Psi}^\dagger(\vec r) \hat{\Psi}^\dagger(\vec r) \hat{\Psi}(\vec r) \hat{\Psi}(\vec r) \; \dd^3 r,
\end{split}
\end{equation}
with the position representation of the single particle Hamilton operator
\begin{equation}
H_\text{sp} (\vec{r},\gamma) = -\frac{\hbar^2}{2m} \nabla^2 + V(\vec{r},\gamma).
\end{equation}
As we consider ultra-cold bosonic atoms the field operator $\hat{\Psi} (\vec{r})$ obeys $[\hat{\Psi}(\vec{r}), \hat{\Psi}^\dagger(\vec{r}')] = \delta (\vec{r}-\vec{r}')$. For arrays of deep traps, it is convenient to expand $\hat{\Psi} (\vec{r})$ using orthogonal atomic orbitals which are localized around the trap minima. For regular lattices the natural choice are Wannier functions $w_{i}^n (\vec{r},\gamma)$ for the $i^\text{th}$ trap site and the $n^\text{th}$ band \cite{Kohn1959, Walters2013, Pichler2013} with the 
corresponding quantized amplitudes
\begin{equation}
\label{eq:decomposition}
\hat{a}_{i}^n(\gamma) = 
\int w_i^{n} (\vec{r},\gamma) \hat{\Psi} (\vec{r}) \; \dd^3 r.
\end{equation}
In order to have a compact notation, we suppress the parameter-dependence if unambiguous. From Eqs. \eqref{eq:many_body_H} and \eqref{eq:decomposition} one obtains the multi-band Bose-Hubbard Hamilton operator \cite{Dutta2015}
\begin{equation}
\label{eq:multibandBHM}
\begin{split}
\hat{H}(\gamma) = 
&\sum\limits_{n, i} \epsilon_i^n (\gamma) 
\hat{a}_{i}^{n \dagger} \hat{a}_{i}^{n \phantom{\dagger}}
- \sum\limits_{n, i \neq j} 
J_{ij}^{n}(\gamma) 
\hat{a}_{i}^{n \dagger} \hat{a}_{j}^{n \phantom{\dagger}}\\
& + \frac{1}{2} \sum\limits_{nopq} \sum\limits_{ijkl} 
U_{ijkl}^{nopq} (\gamma) 
\hat{a}_{i}^{n \dagger} \hat{a}_{j}^{o \dagger} 
\hat{a}_{k}^{p \phantom{\dagger}} \hat{a}_{l}^{q \phantom{\dagger}},
\end{split}
\end{equation}
with on-site energies $\epsilon_i^n$, tunneling parameters $J_{ij}^{n}$, and interaction strengths $U_{ijkl}^{nopq}$ given by
\begin{align}
\epsilon_i^n (\gamma) &= \int w_i^{n} (\vec{r}) H_\text{sp}(\vec{r},\gamma) w_i^{n } (\vec{r}) \; \dd^3r,\\
\label{eq:J}
J_{ij}^n (\gamma) &= -\int w_i^{n} (\vec{r}) H_\text{sp}(\vec{r},\gamma) w_j^{n } (\vec{r}) \; \dd^3r,\\
U_{ijkl}^{nopq} (\gamma) &= g \int w_i^{n} (\vec{r})  w_j^{o} (\vec{r})  w_k^{p} (\vec{r})  w_l^{q} (\vec{r}) \; \dd^3 r.
\label{eq:U}
\end{align}

An arbitrary state $\ket{\psi}$ can be expanded in the Fock basis
\begin{align}
\label{eq:occupancy1}
\ket{\psi} &= \sum\limits_{|\vec{\eta}|=N}
\psi_{\vec{\eta}} \ket{\vec{\eta}},
&\vec{\eta}= ( \eta_1^0, \ldots ,\eta_M^0, \eta_1^1, \eta_2^1, 
\ldots).
\end{align}
Here, $\eta_i^n\in \mathds{N}_0$ is the occupation of the Wannier mode corresponding to the $n^\text{th}$ band and the $i^\text{th}$ site. The set $\mathds{N}_0$ includes the natural numbers and zero. The occupations $\eta_i^n$ are constrained to the number of atoms \mbox{$N=|\vec{\eta}| \equiv \sum_{n,i} \eta^n_i$}.

\subsection{Atomic limit}
\label{sec:adiabHO}

\begin{figure}
	\centering
	\includegraphics[scale=1]{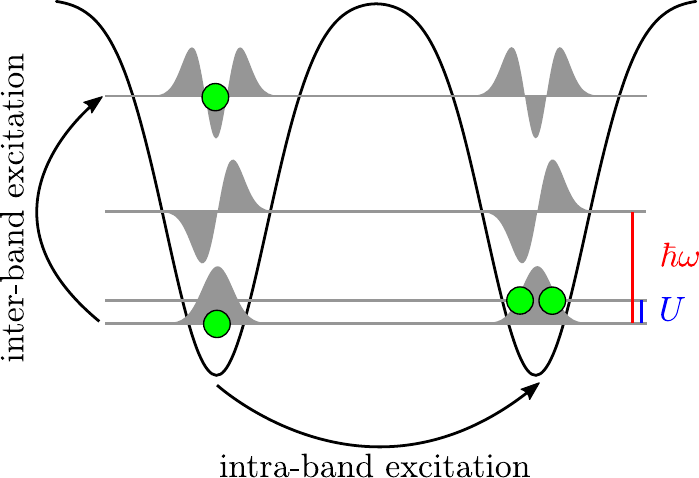}
	\caption{Excitation pathways in a microtrap array:	inter-band excitations dominate in deep traps since intra-band tunneling is exponentially suppressed. However, for shallower potentials	intra-band tunneling prevails, as long as the two-particle interaction energy $U \equiv U_{iiii}^{0000}$ remains smaller than the band gap $\hbar \omega$.}
	\label{fig:excitations}
\end{figure}

The transfer process starts from an array of tight isolated traps with one atom per site prepared in the respective motional ground state. The corresponding many-body state is given by $\ket{g_\text{al}} = 
\ket{\vec{\eta}}$ with $\eta_i^n = \delta_{0n}$. This regime is called the \emph{atomic limit}, where inter-site tunneling is strongly suppressed. Therefore, the only possible reaction of the system to time-dependent modulations of the trap depth are local inter-band excitations (cf. \fig{fig:excitations}) resulting in states of the form $\hat{a}_i^{n\dagger} \hat{a}_i^0 \ket{g_\text{al}}$. 

Due to the tight confinement of the atoms around the respective potential minima, each trap can be described by a harmonic oscillator. The corresponding frequencies $(\Omega_{x}(t), \Omega_{y}(t), \Omega_{z}(t)) = \gamma(t)$ are the control parameter for the adiabatic transfer procedure. The multi-band Bose-Hubbard Hamilton operator of \eq{eq:multibandBHM} reduces to the sum of local harmonic oscillators
\begin{align}
\hat{H}_\text{al}(\gamma) &= \sum\limits_{n,i} \epsilon_i^{n}(\gamma) \; \hat{a}_{i}^{n \dagger}  
\hat{a}_{i}^{n \phantom{\dagger}},\\
\epsilon_i^{n}(\gamma) &= \hbar 
\sum\limits_{l=x,y,z} (n_l+\tfrac{1}{2}) \Omega_l.
\end{align}
If we introduce local Cartesian coordinates $\vec{\xi}=\vec{r}-\vec{R}_i$ around the trap minimum $\vec{R}_i$ of the $i^\text{th}$ site, then the Wannier function 
\begin{align}
w_i^{n} (\vec{r})=w^{n} (\vec{\xi})= (\vec{\xi} |n_x n_y n_z ),
\end{align}
factorizes into one-dimensional harmonic oscillator states
\begin{align}
\label{eq:wannierHO}
( \xi_l|n_l ) &= \frac{e^{-\frac{\xi_l^2}{2 a_l^2}}}{\sqrt[4]{ \pi (2^{n_l} n_l! a_l)^2}}
H_{n_l} \left(\tfrac{\xi_l}{a_l}\right) .
\end{align}
Here, $a_l=\sqrt{\hbar/(m \Omega_l)}$ denote the three oscillator lengths, $n = (n_x,n_y,n_z) \in \mathbb{N}_0^3$ are the motional quantum numbers, and $H_m$ is the $m^\text{th}$ Hermite polynomial.

In order to determine the adiabatic Lagrangian function $\mathcal{L}_\text{al}(\gamma,\dot{\gamma})$ from \eq{eq:adiasum}, parameter derivatives of the form
\begin{equation}
	\label{eq:H_derivative}
	\frac{\partial \hat{H}_{\text{al}}}{\partial \gamma_l} = 
	\sum\limits_{n,i} \frac{\partial \epsilon_i^n}{\partial \gamma_l} \hat{a}_{i}^{n \dagger}  
	\hat{a}_{i}^{n \phantom{\dagger}} 
	+ \epsilon_i^n \bigg( \frac{\partial \hat{a}^{n\dagger}_{i}}{\partial \gamma_l} 
	\hat{a}_{i}^n + \hat{a}^{n\dagger}_{i} \frac{\partial \hat{a}_{i}^n}{\partial \gamma_l} \bigg)
\end{equation}
need to be calculated. The derivatives of the operators $\hat{a}_{i}^n$ can be found from \eq{eq:decomposition} \cite{Pichler2013, Lacki2013}
\begin{align}
\frac{\partial \hat{a}_{i}^n}{\partial {\gamma_l} } &= \sum\limits_{p,j} C_{ij;l}^{np} \; \hat{a}_{j}^p,&
C_{ij;l}^{np} &= \int \frac{\partial w_i^{n} (\vec{r})}{\partial \gamma_l} w_j^{p} (\vec{r}) \; \dd^3r.
\label{eq:a_derivative}
\end{align}
The coefficients $C_{ij;l}^{np}$ can be interpreted geometrically as the generators of a basis-rotation and satisfy the relation $C_{ij;l}^{np}=-{C_{ji;l}^{pn}}$. Using the harmonic approximation for the Wannier functions given in \eq{eq:wannierHO}, we obtain
\begin{align}
C_{ii;l}^{n0} &
= \frac{\delta_{n_l 2}}{\sqrt{8} \gamma_l} 
\prod\limits_{l' \neq l}^\ell
 \delta_{n_{l'} 0}.
\end{align}
The calculation of the transition amplitudes defined in \eq{rotspeed} requires evaluation of  the matrix-element between the ground and excited state. Using Eqs.~\eqref{eq:H_derivative} and \eqref{eq:a_derivative} we find
\begin{align} 
\bra{g_\text{al}} \hat{a}_i^{0\dagger} \hat{a}_i^n
\partial_{\gamma_l} \hat{H}_\text{al} \ket{g_\text{al}}
&= (\epsilon_i^n -\epsilon_i^0) 
C_{ii;l}^{n0},
\end{align}
yielding
\begin{align} 
\alpha_{ii}^{0n} 
&=\sum_{l=1}^3
\dot{\gamma}_l (\epsilon_i^n -\epsilon_i^0) 
C_{ii;l}^{n0}.
\end{align}
The energy of inter-band excitations 
\begin{equation}
\label{eq:energiesHO}
\hbar \omega_{i}^{n} (\gamma) = 
\epsilon_i^n - \epsilon_i^0=
\hbar \sum\limits_{l=1}^3 n_l \gamma_l,
\end{equation}
can be inferred from the harmonic oscillator level spacing. Finally, by summing over all excited states, we determine the adiabatic error Lagrangian function in the atomic limit
\begin{align}
\label{eq:adiadeep}
\mathcal{L}_\text{al}(\gamma,\dot{\gamma}) &
=\sum\limits_{l=1}^3
\frac{1}{2}
\mathcal{M}_\text{al}(\gamma_l) {\dot{\gamma}_l}^2,
\end{align}
with the extensive mass-function 
$\mathcal{M}_\text{al}(\gamma_l)= M (2\gamma_l)^{-4}$ and the number of sites $M$.

Fortunately, $\mathcal{L}_\text{al}$ is separable. Due to the integrability condition of \eq{integral}, we obtain the optimal adiabatic ramp with the well-known hyperbolic shape \cite{Chen2010} 
\begin{equation}
\label{eq:horamp}
\gamma_l^{-1}(t) =
\Omega_{l0}^{-1} + 
\left(\Omega_{l\tau}^{-1} -\Omega_{l0}^{-1}\right) 
\frac{t}{\tau},
\end{equation}
for the transfer of trapped particles from an initial trap at $t=0$ with $\gamma (0) = (\Omega_{x0},\Omega_{y0},\Omega_{z0})$ to a final trap at $\tau$ with 
$\gamma (\tau) = (\Omega_{x\tau},\Omega_{y\tau},\Omega_{z\tau})$. 
The quantitative measure for residual excitations 
\begin{equation}
\mathds{E}_\infty^\text{al} [\gamma , \dot{\gamma}] =
\frac{M}{32 \tau^2}\sum_{l=x,y,z} 
{(\Omega_{l\tau}^{-1}-\Omega_{l0}^{-1})^2}
\label{eq:EinfAL}
\end{equation}
is inversely proportional to the square of the ramp duration $\tau$.

In an experiment the trap frequencies are determined by the optical potential. Therefore, the actual control parameter is the trap depth. In Sec.~\ref{sec:derivedparam} the relations between the trap frequencies and the trap depth are derived for realistic system parameters obtained from experiments. 

\subsection{Mott insulator}
\label{adiabMI}

For shallower traps, one obtains an itinerant many-body state. In this regime intra-band excitations due to tunneling between adjacent traps (cf. \fig{fig:excitations}) are energetically favored over inter-band excitations. We assume that the initial cooling process was efficient and the preceding adiabatic transfer has 
not populated higher bands. Therefore, we restrict the following analysis to the lowest band. We further assume sufficiently deep traps such that only nearest-neighbor tunneling and on-site interactions need to be considered and that the trap array is homogeneous $\epsilon\equiv\epsilon_i^0, J\equiv J_{i,i+1}^0,
U\equiv U^{0000}_{iiii}$. In this case, the single-band Bose-Hubbard model \cite{Hubbard1963, Fisher1989, Jaksch1998, Bloch2008, Lewenstein2012}
\begin{align}
\label{eq:HJ}
\begin{split}
\hat{H}_\text{bh}(\gamma)  &= 
\epsilon N
-J \sum\limits_{\langle i,j\rangle} 
\hat{a}_i^{0 \dagger} \hat{a}_j^{0\phantom{\dagger}}
+\frac{U}{2} \sum\limits_i 
\hat{a}^{0 \dagger}_i 
\hat{a}_i^{0 \dagger}
\hat{a}_i^{0\phantom \dagger} 
\hat{a}_i^{0\phantom \dagger}
\end{split}
\end{align}
emerges from \eq{eq:multibandBHM}. The notation $\langle i,j \rangle$ indicates a summation over 
nearest-neighbor pairs of traps. The relevant control parameter is $\gamma = (J,U)$, since the on-site single-particle energy $\epsilon$ results only in a constant energy offset.

In order to evaluate the adiabatic Lagrangian function from \eq{eq:adiasum}, one needs to find the energy eigenstates of $\hat{H}_\text{bh}$. For $U \gg J$, this can be done perturbatively starting from the ground state in the atomic limit $\ket{g_\text{al}}$ \cite{Freericks1994}. In the Mott-insulator phase, low lying excited states  $\ket{p,q} = \hat{a}_p^{0\dagger} \hat{a}_q^0 \ket{g_\text{al}}/\sqrt{2}$, transport an atom from site $q$ to an occupied site $p \neq q$. These transitions are called particle-hole or intra-band excitations (cf. \fig{fig:excitations}). To first order in perturbation theory, the ground state reads
\begin{equation}
\ket{g_\text{bh}} = \ket{g_\text{al}} + \frac{\sqrt{2}J}{U} \sum\limits_{
\langle p,q \rangle} \ket{p,q} +\mathcal{O} (\tfrac{J^2}{U^2}).
\end{equation}
The energy corresponding to a particle-hole excitation is given by
\begin{align}
\hbar \omega_{pq} &= \bra{p,q} \hat{H}_\text{bh} \ket{p,q} - \bra{g_\text{bh}} \hat{H}_\text{bh} \ket{g_\text{bh}}\\
 &= U + \mathcal{O} (\tfrac{J^2}{U^2}).
\label{eq:Hubbard_freq}
\end{align}
The transition matrix elements can be calculated from equation \eq{rotspeed} 
yielding
\begin{align}
\begin{split}
\alpha_{pq} &= \dot{U} \bra{p,q} \partial_U \hat{H}_\text{bh}\ket{g_\text{bh
}} + \dot{J} \bra{p,q} \partial_J \hat{H}_\text{bh}\ket{g_\text{bh}}\\ 
&= 
-\sqrt{2} U \partial_t \left(\frac{J}{U}\right)  + \mathcal{O} (\tfrac{J^2}{U^
2}).
\end{split}
\label{eq:Hubbard_alpha}
\end{align}
It is worth noting that a change in the parameters $U$ and $J$ is connected to a change in the Wannier functions. Therefore, the derivative of the operators $\hat{a}_i$ with respect to $U$ and $J$ need to be considered. However, terms connected to these derivatives are neglected in \eq{eq:Hubbard_alpha} since they do not induce intra-band excitations \cite{Lacki2013a, Pichler2013}.

From Eqs. \eqref{eq:Hubbard_freq} and \eqref{eq:Hubbard_alpha} the adiabatic functional on a two-dimensional parameter space 
$\gamma=(J,U)$ can be derived
\begin{align}
\mathcal{L}_\text{bh}(\gamma,\dot{\gamma}) 
&= \frac{1}{2}\sum_{k,l=1}^{2}{\dot{\gamma}_k}\mathcal{M}_{kl}(\gamma)
\dot{\gamma}_l+\mathcal{O} (\tfrac{J^3}{U^3}),\\
 \mathcal{M}&=
\frac{4M z \hbar^2}{U^6}
\begin{pmatrix}
	U^2&-J U\\
	-J U & J^2
\end{pmatrix},
\end{align}
with $z$ being the average number of nearest-neighbor sites, commonly called coordination number. 

In experiments \cite{Bloch2008}, the on-site interaction strength $U(t)=U(\mathcal{V}(t))$ and the tunneling parameter $J(t)=J(\mathcal{V}(t))$ are not independent variables but functionally depend on the depth of the optical potential $\mathcal{V}(t)$. This is described in Sec.~\ref{sec:derivedparam}. Therefore, we obtain a one-dimensional parameter curve $\gamma(t)=\mathcal{V}(t)$ and adiabatic Lagrangian function
\begin{align}
\label{eq:Lbh}
\mathcal{L}_\text{bh}(\gamma,\dot{\gamma})&=
\frac{1}{2}\mathcal{M}_\text{bh}(\gamma)\dot{\gamma}^2,\\
\mathcal{M}_\text{bh} (\gamma) &=\frac{4 M z \hbar^2}{U^2(\gamma)}
\left[\partial_\gamma \left(\frac{J(\gamma)}{U(\gamma)}\right)\right]^2,
\label{eq:Mbh}
\end{align}
with a well-defined positive mass function $\mathcal{M}_\text{bh}(\gamma)>0$.

\subsection{Additivity of errors}
Clearly, the transfer crosses from the atomic limit to the Mott insulator limit. The full description of the dynamics between the two extreme limits is very complex, because inter- as well as intraband excitation become relevant simultaneously. 
Therefore, we propose as an approximate measure for the 
instantaneous adiabaticity of the transfer process 
the sum of errors
\begin{equation}
\label{eq:additivity}
\mathcal{L} (\gamma, \dot{\gamma}) = \mathcal{L}_\text{al} (\gamma, \dot{\gamma}) + \mathcal{L}_\text{bh} (\gamma, \dot{\gamma}).
\end{equation}
This additivity of errors follows directly from \eq{eq:adiasum}, which measures the error $\mathcal{L}$ by a sum of squares. However, we have approximated the individual terms by using the limiting expression derived in previous sections.

\section{Realistic experimental setting}
\label{sec:example}

In this section we discuss details of an implementation based on recent experiments. From this we determine a realistic set of experimental parameters and derive relations for trap frequencies, interaction strengths, and tunneling parameters.

\subsection{Optical potential}
\label{sec:realisticp}

There are multiple techniques to generate arrays of optical microtraps. Among these are acousto-optic deflectors (AODs) \cite{Kaufman2014, Lester2015, Endres2016}, spatial light modulators (SLMs) \cite{Nogrette2014, Barredo2016}, and microlens arrays (MLAs) \cite{Birkl2001, Schlosser2011, Sturm2017}. Here, we make no assumptions about the used approach. However, we presume that the microtraps have an approximately Gaussian shape with a waist of $w_0$~=~\unit{0.71}{\micro m} and are generated by linearly polarized light with a wavelength of $\lambda_\bot$~=~\unit{852}{nm} as in \cite{Kaufman2014}. Further, we consider the species $^{87}$Rb, which is the workhorse for the field of ultracold atoms and has been used in most of the experiments relevant to this work, e.g. \cite{Kaufman2012, Kaufman2014, Lester2015, Endres2016, Barredo2016}. We assume that the atoms are prepared in the state $\ket{5^2 S_{1/2},F=2,m_F=2}$ as they were in \cite{Kaufman2012, Kaufman2014}. In \cite{Lester2015} the setup from \cite{Kaufman2014} has been used to generated a $2 \times 2$ optical tweezer array with one atom per trap. The minimal trap spacing that allows for a high preparation efficiency of $90 \%$ has been determined to $d$~=~\unit{1.7}{\micro m}. For this trap spacing the overlap of adjacent traps is negligible, which facilitates cross-talk free single-site control over the optical potential \cite{Sturm2017}.
\begin{table}
\centering
\begin{ruledtabular}
\begin{tabular}{l c r}
Quantity & Symbol & Value \\
\hline
Atomic mass of $^{87}$Rb & $m$ & \unit{86.9}{u}\\
Scattering length & $a_s$ & \unit{5.24}{nm}\\
Energy scale & $\mathcal{E}$ & $\unit{38.1}{nK} \cdot k_B$\\
Wavelength for $V_\bot$ & $\lambda_\bot$ & \unit{852}{nm} \\
Wavelength for $V_\parallel$ & $\lambda_\parallel$ & \unit{1064}{nm}\\
Trap spacing & $d$ & \unit{1700}{\nano m}\\
Trap waist & $w_0$ & \unit{710}{\nano m}\\
Inclination angle & $\theta$ & \unit{24.6}{\degree}\\
Initial depth of $V_\bot$ & $\mathcal{V}_\bot(0)$ & $\unit{1}{mK} \cdot k_B$\\
Final depth of $V_\bot$ & $\mathcal{V}_\bot(\tau)$ & $\unit{158}{nK} \cdot k_B$\\
Initial depth of $V_\parallel$ &$\mathcal{V}_\parallel(0)$ & $\unit{2.5}{mK} \cdot k_B$\\
Final depth of $V_\parallel$ & $\mathcal{V}_\parallel(\tau)$ & $\unit{395}{nK} \cdot k_B$
\end{tabular}
\end{ruledtabular}
\caption{Experimental parameters used for obtaining realistic estimates for the adiabatic transfer procedure.}
\label{table:parameters}
\end{table}

For the experiments in \cite{Kaufman2012, Thompson2013, Kaufman2014} the cooling efficiency in axial direction was considerably lower than in the transverse direction. This results from weaker confinement in the axial direction. The effect can be compensated by additional axial confinement. Further, this prevents atoms from tunneling to diffraction patterns along the optical axis that exist if the trap array is generated by a MLA or a SLM (cf. the Talbot effect). Therefore, we consider axial confinement implemented by a standing wave which is produced by two laser beams with a wavelength $\lambda_\parallel$~=~\unit{1064}{nm} that enclose an angle of $\theta = 24.6^\circ$. This results in a spacing of \unit{2.5}{\micro m} between the antinodes of the optical potential which is large enough to prohibit tunneling in axial direction for the considered potential depths. The total optical potential reads
\begin{equation}
V(\vec{r},t) = V_\bot (\vec{r},t) + V_\parallel (\vec{r},t)
\end{equation}
with the optical microtrap array potential
\begin{equation}
\label{eq:mtarray}
V_\bot(\vec{r},t) \approx - \sum\limits_{i=1}^N \mathcal{V}_\bot (t) e^{-2 \tfrac{(x-X_i)^2+(y-Y_i)^2}{w_0^2}}
\end{equation}
and the standing wave potential for axial confinement
\begin{equation}
\label{eq:standingwave}
V_\parallel (\vec{r},t) \approx - \mathcal{V}_\parallel (t) \cos^2 (\kappa z ).
\end{equation}
Here, we introduced the potential depths $\mathcal{V}_\bot$ and $\mathcal{V}_\parallel$ as well as the $i^\text{th}$ site's coordinates $X_i$ and $Y_i$. The projection of the wave vector onto the lattice direction $\kappa = \sin (\theta/2) \, 2 \pi/\lambda_\parallel$ determines the periodicity of the 1D optical lattice used for axial confinement. For Eqs.~\eqref{eq:mtarray} and \eqref{eq:standingwave} it is assumed that the 
out-of-plane confinement from $V_\bot$ is weak in comparison to that from $V_\parallel$ and that the laser beams generating $V_\parallel$ have a waist that is larger than the extent of the microtrap array. During the cooling process we assume $\mathcal{V}_\bot / k_B = \unit{1}{mK}$, which is consistent with the values used in experiments \cite{Kaufman2012, Thompson2013, Kaufman2014}. In order to have an equally strong confinement in the out-of-plane direction we choose $\mathcal{V}_\parallel/k_B = \unit{2.5}{mK}$. The chosen parameters are summarized in Table \ref{table:parameters}.

\subsection{Trap frequencies and Bose-Hubbard parameters}
\label{sec:derivedparam}

In order to evaluate the expressions for the adiabatic Lagrangian functions derived in \sect{sec:adiab} we need to express the trap frequencies and the Hubbard parameters as functions of the optical potential depths $\mathcal{V}_\bot$ and $\mathcal{V}_\parallel$. This will be done in the present subsection. The harmonic trapping frequencies can be computed from the curvature of the potentials given in Eqs.~\eqref{eq:mtarray} and \eqref{eq:standingwave} yielding
\begin{align}
\label{eq:freq1}
\Omega_{x} &= \Omega_{y} =\sqrt{\frac{4 \mathcal{V}_{\bot}}{m w_0^2}},
&\Omega_z &= \sqrt{\frac{2 \kappa^2 \mathcal{V}_{\parallel}}{m}}.
\end{align}
In combination with \eq{eq:adiadeep}, these expressions allow us to estimate the adiabaticity of the transfer process in the atomic limit.

In order to obtain the Hubbard parameters for tunneling $J$ and on-site interaction $U$, we need to compute the Wannier functions $w_i$. Since the optical potential is a sum of the in-plane part $V_\bot $ and the axial part $V_\parallel$, the Wannier functions factorize
\begin{equation}
\label{eq:Wannier}
w_i^0(\vec{r}) = \varphi_i (x,y) \phi (z).
\end{equation}
In axial direction the tunneling is strongly suppressed at all times. Therefore, a natural choice for $\phi$ is the ground state of one slice of the standing wave potential given in \eq{eq:standingwave}. We calculate $\phi$ by solving the corresponding 1D time-independent Schrödinger equation numerically. For the potential in the $x$-$y$-plane we assume a regular square lattice of $20 \times 20$ sites and periodic boundary conditions. $\varphi_i$ is the lowest band Wannier function for this potential obtained from a numerical band structure calculation \cite{Walters2013}. The Hubbard parameters for tunneling $J$ between adjacent sites 
$i$ and $j$ and the on-site interaction $U$ can be calculated from Eqs. \eqref{eq:J} and \eqref{eq:U} respectively. It is worth noting that the results for $U$ and $J$ can be adopted for different lattice geometries, such as one-dimensional, triangular, and hexagonal lattices since the relative deviations are small.

\begin{figure}
	\centering
	\includegraphics[scale=1]{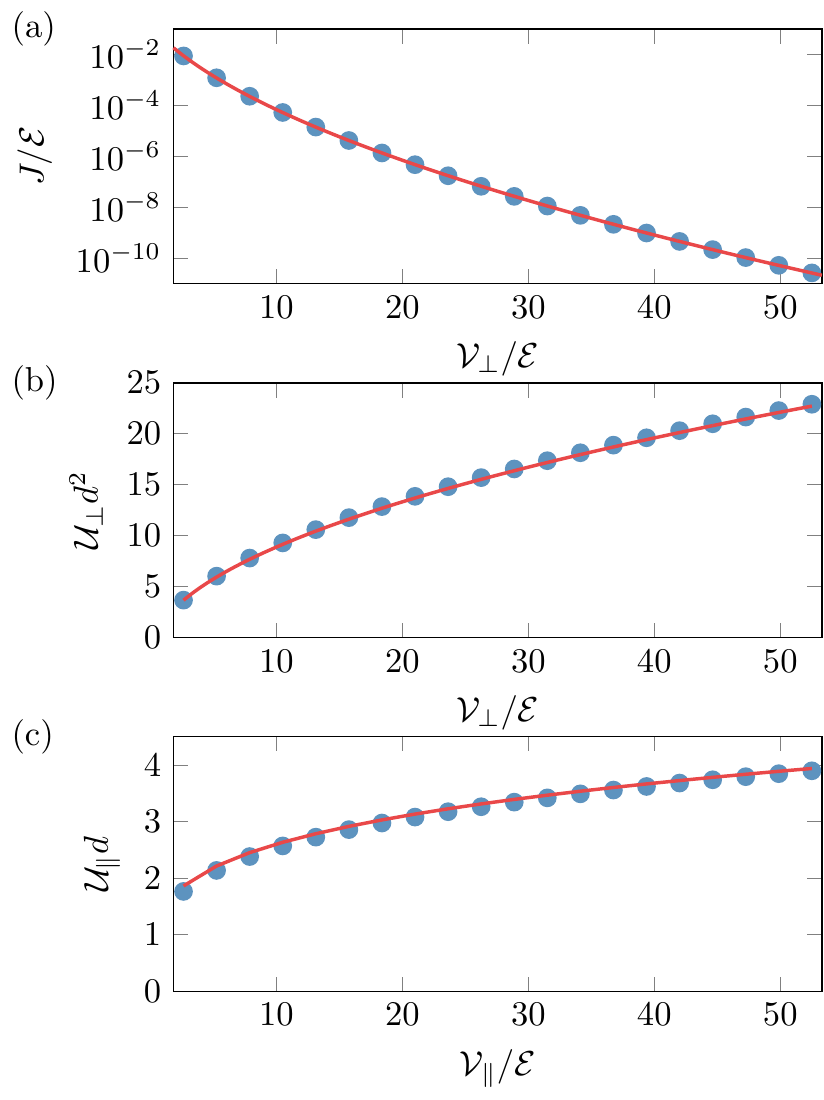}
	\caption{Comparison between results from the numerical solution of the single-particle Schrödinger equation (circles) and the approximate closed form expressions (lines) for the Hubbard parameters. (a) Tunneling parameter $J$ as a function of the potential depth $\mathcal{V}_\bot$ for $^{87}$Rb atoms in a square lattice of Gaussian dipole traps with waist $w_0=\unit{0.71}{\micro m}$ and trap spacing $d=\unit{1.7}{\micro m}$. (b) The in-plane part $\mathcal{U}_\bot$ of the on-site interaction parameter versus 
		$\mathcal{V}_\bot$ for the same parameters as in (a). (c) The out-of-plane part $\mathcal{U}_\parallel$ of the on-site interaction parameter versus $\mathcal{V}_\parallel$ for the same parameters as in (a).}
	\label{fig:BHP}
\end{figure}

For convenience in later computations simple expressions for the Hubbard parameters are advantageous. The on-site interaction strength can be reliably approximated by using Gaussian wave functions for $\phi$ and $\varphi_i$ \cite{Zwerger2003}. For the out-of-plane direction the harmonic oscillator length $a_\parallel$ can be used as $1/\sqrt{e}$ width for the Gaussian wave function
\begin{align}
\label{eq:aparallel}
a_\parallel &= \sqrt[4]{\frac{d^2 \mathcal{E}}{4 \pi^2 \kappa^2 \mathcal{V}_\parallel}}.
\end{align}
Here, we have introduced the natural energy scale of a lattice $\mathcal{E}= h^2/(2md^2)$ with trap spacing $d$. For $^{87}$Rb and $d=\unit{1.7}{\micro m}$ this yields $\mathcal{E} = k_B \, \unit{38.1}{nK} = h \, \unit{794}{Hz}$. In order to obtain a satisfying approximation for the Wannier function $\varphi_i$, we perform a variational calculation to find the wave function's width that minimizes the energy in a Gaussian potential well (cf. Appendix \ref{sec:variationcalc}). This yields
\begin{align}
a_\bot = \sqrt{\frac{w_0^2 d}{2 \pi w_0 \sqrt{2\mathcal{V}_\bot /\mathcal{E}}-2d}}.
\label{eq:abot}
\end{align}

One can evaluate the on-site interaction  given in \eq{eq:U},  
using the above expressions and the factorization ansatz for the Wannier function of \eq{eq:Wannier} and obtains
\begin{equation}
U = g \; \mathcal{U}_\bot \; \mathcal{U}_\parallel ,
\label{eq:U2}
\end{equation}
with the in-plane and axial part defined by
\begin{align}
\label{eq:uinplane}
\mathcal{U}_\bot &= \iint \varphi_i^4 (x,y) \; \dd x \; \dd y \approx \sqrt{\frac{2 \mathcal{V}_\bot}{w_0^2 d^2 \mathcal{E}}} -\frac{1}{\pi w_0^2}  ,\\
\mathcal{U}_\parallel &= \int \phi^4 (z) \; \dd z  \approx \sqrt[4]{\frac{\kappa^2 \mathcal{V}_\parallel}{d^2\mathcal{E}}}.
\label{eq:uooplane}
\end{align}

The tunneling parameter $J$ can not be well approximated using the Gaussian wave function ansatz because it significantly underestimates the Wannier function's value at the position of neighboring sites. Instead, we parametrize $J$ using a semiclassical ansatz \cite{Gerbier2005, Lee2009}
\begin{equation}
\label{eq:Jfit}
J = A \, (\tfrac{\mathcal{V}_\bot}{\mathcal{E}})^C \, e^{ -B \sqrt{\mathcal{V}_{\bot}/\mathcal{E}}} \mathcal{E}.
\end{equation}
A fit to our numerical calculations yields $A=2.26\pm0.05$, $B=4.02\pm0.01$, $C=1.00\pm0.03$. \fig{fig:BHP} shows the comparison between the discussed approximations and the results from the numerical band structure calculations revealing quantitative agreement.

\section{Rapid adiabatic preparation of a Mott insulator}
\label{sec:mott}
In this section the transfer from the atomic limit to a Mott insulator close to the quantum phase transition is investigated. The challenge is to find ramps $\mathcal{V}_\bot (t)$ and $\mathcal{V}_\parallel (t)$ that minimize excitations during this process. 
This is resolved by minimizing the functional Eq.~\eqref{eq:measure1} with the error measure Eq.~\eqref{eq:additivity}.

Before optimal ramp shapes can be computed the initial and final values for the potential depth need to be determined. The initial values are fixed by the requirement of efficient sideband cooling and given in \sect{sec:realisticp}, whereas the final values are determined by the targeted many-body regime. In this case we want to prepare the system in the Mott-insulator phase close to the phase transition, occurring at $U/J=3.4$ for a 1D lattice. Therefore, we choose a final value of $U/J=10$. 

In order to obtain equal trap frequencies in all directions $\Omega_x=\Omega_y=\Omega_z$, we choose a constant ratio 
\begin{equation}
	\frac{\mathcal{V}_\bot(t)}{\mathcal{V}_\parallel (t)} = \frac{\kappa^2 w_{0\bot}^2}{2}.
\end{equation}
This determines the final potential depths $\mathcal{V}_\bot (\tau)/k_B = \unit{158}{nK}$ and $\mathcal{V}_\parallel (\tau)/k_B = \unit{395}{nK}$ yielding $U/h$~=~\unit{22}{Hz} and $J/h$~=~\unit{2.2}{Hz}. Due to the constant ratio between the potential depths the instantaneous adiabatic Lagrangian function $\mathcal{L}$ can be expressed as a function of $ \mathcal{V}_\bot$ and $\dot{\mathcal{V}}_\bot$ only.

\subsection{Optimal ramps for the potential depth}
\label{sec:ramps}

To gauge the quality of the procedure of Sec.~\ref{sec:adiab}, we chose a standard approach for finding an  adiabatic ramp $\mathcal{V}_\bot (t)$ as a reference. We use a suitable set of test functions as an ansatz and optimize the parameters. Since the system traverses two different regimes, which are associated with two different time scales for an adiabatic transfer (cf. \sect{sec:adiab}), we choose a bi-exponential ansatz of the form
\begin{align}
\label{eq:ramp1}
\mathcal{V}_\bot(t) &= V_a e^{-t/\tau_a} + 
V_b e^{-t/\tau_b},
\end{align}
with time constants $\tau_a,\tau_b$ and amplitudes $V_a,V_b$. The amplitudes are fixed by imposing the 
boundary values at $t=0$ and $t=\tau$. The time constants are computed by numerically minimizing the 
quantity $\mathds{E}_\infty$, i.\thinspace{}e. calculating
\begin{equation}
\min\limits_{\tau_a,\tau_b}\mathds{E}_\infty(\mathcal{V}_\bot, \dot{\mathcal{V}}_\bot).
\end{equation}

The red line in \fig{fig:potentialdepth} shows the resulting ramp $\mathcal{V}_\bot (t)$ for the given parameters and $\tau $~=~\unit{50}{ms}. For this ramp the time dependencies of $\mathcal{L}_\text{al}$ and $\mathcal{L}_\text{bh}$ are shown in \fig{fig:adiabatictransfer} (dotted and solid red line respectively). The fact that during the first \unit{15}{ms} both $\mathcal{L}_\text{al}$ and $\mathcal{L}_\text{bh}$ are much smaller than $\mathds{E}_\infty$ indicates that a better ramp can be realized with a faster decrease during this time interval.

\begin{figure}
	\centering
	\includegraphics[scale=1]{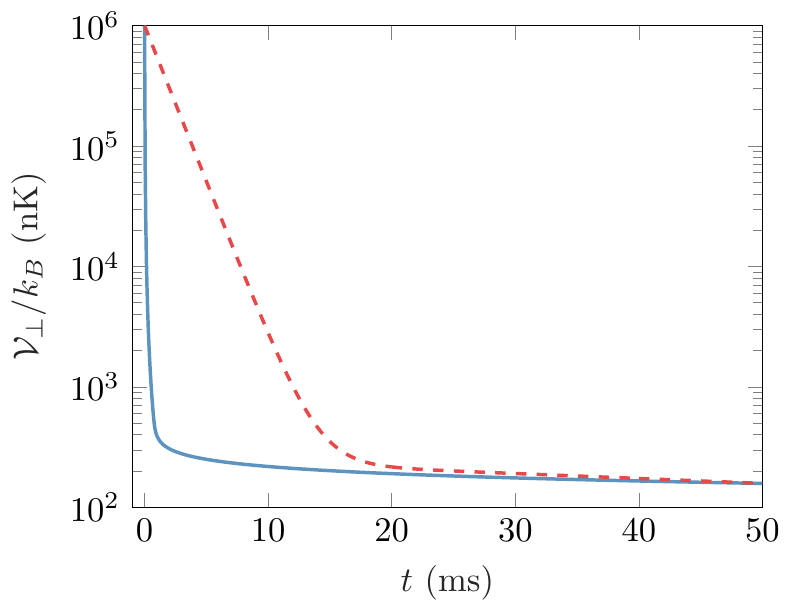}
	\caption{Potential depth $\mathcal{V}_\bot (t)$ versus time $t$. The red dashed line shows the ramp resulting from the bi-exponential ansatz whereas the blue solid line is the optimal adiabatic ramp.}
	\label{fig:potentialdepth}
\end{figure}

\begin{figure}
	\centering
	\includegraphics[scale=1]{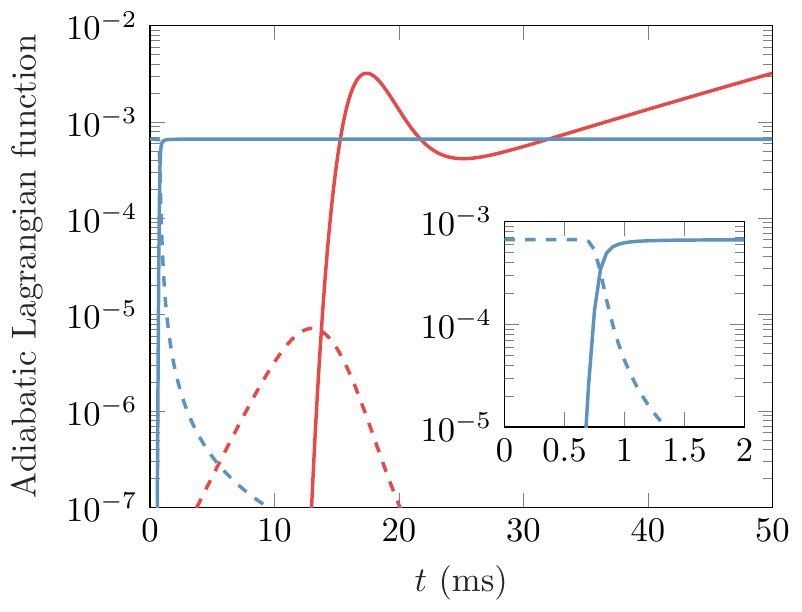}
	\caption{The adiabatic Lagrangian functions per site $\mathcal{L}_\text{al}/M$ (dashed) and $\mathcal{L}_\text{bh}/M$ (solid) are plotted versus time $t$ for a bi-exponential (red) and an optimal adiabatic (blue) transfer sequence of duration $\tau = \unit{50}{ms}$.}
	\label{fig:adiabatictransfer}
\end{figure}

The variational approach proposed in this article follows from solving the Euler-Lagrange equation
\begin{equation}
\label{eq:EL}
\frac{\text{d}}{\text{d}t} \frac{\partial \mathcal{L}}{\partial \dot{\mathcal{V}}_\bot}= \frac{\partial \mathcal{L}}{\partial \mathcal{V}_\bot},
\end{equation}
as discussed in \sect{sec:adiab}. The explicit form of the above equation can be obtained by using 
\Cref{eq:adiadeep,eq:Lbh,eq:Mbh,eq:freq1,eq:aparallel,eq:abot,eq:U2,eq:uinplane,eq:uooplane,eq:Jfit,eq:additivity}. In general, a solution $\mathcal{V}_\bot (t)$ to the above equation makes the functional $\mathds{E}_1$ stationary. However, in Appendix \ref{sec:proof} we show that in this particular case it also minimizes $\mathds{E}_1$ and $\mathds{E}_\infty$. Therefore, a solution to \eq{eq:EL} can be considered as an optimal adiabatic ramp. It is worth noting that $\mathcal{L}$ is a constant of motion. Therefore, the optimal adiabatic ramp is equivalent to constant adiabaticity pulses used in nuclear magnetic resonance \cite{Baum1985}. For a ramp duration of $\tau$~=~\unit{50}{ms}, this ramp is shown in \fig{fig:potentialdepth} (blue line). 
As expected from the discussion of the bi-exponential ramp 
function the optimal ramp shape shows a much faster decrease 
for $t<$~\unit{15}{ms}. The dashed and solid blue lines in \fig{fig:adiabatictransfer} show the time dependence of the 
components $\mathcal{L}_\text{al}$ and $\mathcal{L}_\text{bh}$ respectively. 
This demonstrates that inter-band excitations are only relevant during the first millisecond. Thereafter, intra-band excitations dominate. 
The transition point between the regimes is 
given by the condition $\mathcal{M}_{\text{al}}(\mathcal{V}_\bot^c)=
\mathcal{M}_{\text{bh}}(\mathcal{V}_\bot^c)$ and specifies a characteristic value of the parameter 
$\mathcal{V}_\bot^c/k_B=\unit{479}{\nano \kelvin}$.

In the following, we derive analytic 
expressions for the optimal adiabatic ramp shapes
in the regimes of intra-band and inter-band excitations, respectively.
It is straight forward to obtain the ramp shape for the initial time interval, in which inter-band excitations dominate, using Eqs. \eqref{eq:horamp} and \eqref{eq:freq1}
\begin{equation}
\mathcal{V}_\bot (t) = 
\frac{\mathcal{V}_1 \mathcal{V}_2 }{
\left[ \sqrt{\mathcal{V}_2} + 
\left( \sqrt{\mathcal{V}_1} -  \sqrt{\mathcal{V}_2} \right) \frac{t}{\tau_1} \right]^2}, \quad \forall t<\tau_1.
\label{eq:analytic_al}
\end{equation}
Here, we have introduced $\mathcal{V}_1 = \mathcal{V}_\bot (0)$, 
$\mathcal{V}_2 = \mathcal{V}_\bot (\tau_1)$, and 
$\tau_1 = \unit{0.7}{ms}$ which marks the end of the first time interval (cf. inset of Fig. \ref{fig:adiabatictransfer}). 

\begin{figure}
	\centering
	\includegraphics[scale=1]{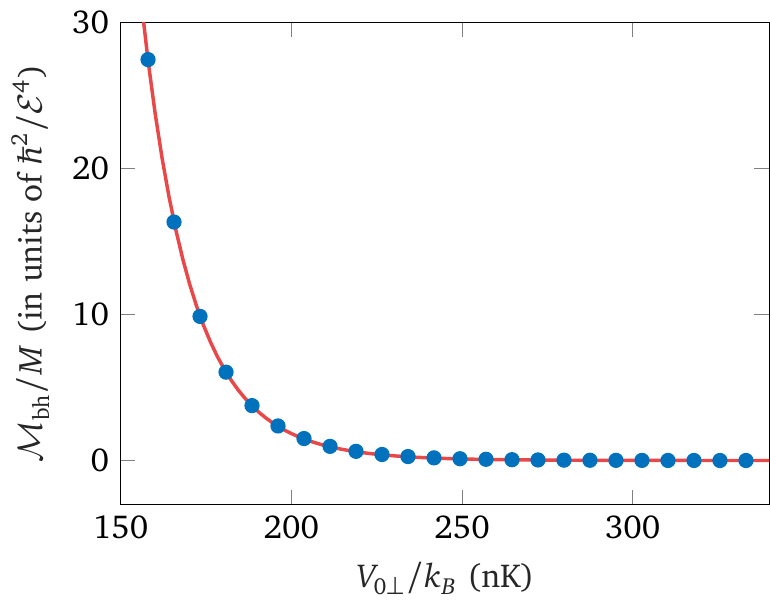}
	\caption{The mass function per site $\mathcal{M}_\text{bh}/M$ is plotted versus the potential depth $\mathcal{V}_\bot$. The blue points represent the full mass function obtained from Eqs. \eqref{eq:Lbh}, \eqref{eq:Mbh}, \eqref{eq:U2}, \eqref{eq:uinplane}, \eqref{eq:uooplane}, and \eqref{eq:Jfit}, whereas the orange line corresponds to the approximation given in \eq{eq:massfunction}.}
	\label{fig:massfunction}
\end{figure}

For the second time interval intra-band transitions dominate. The corresponding mass function $\mathcal{M}_\text{bh}$ can be determined from Eqs. \eqref{eq:Lbh}, \eqref{eq:Mbh}, \eqref{eq:U2}, \eqref{eq:uinplane}, \eqref{eq:uooplane}, and \eqref{eq:Jfit}. This complicated expression prohibits an analytic calculation of the integral in \eq{integral}. However, for the relevant parameter regime we find that
\begin{equation}
\mathcal{M}_\text{bh} (\mathcal{V}_\bot) \approx \frac{M\hbar^2}{\mathcal{E}^3\mathcal{V}_\bot} \exp \left(a - b \sqrt{\mathcal{V}_\bot} \right).
\label{eq:massfunction}
\end{equation}
with fit parameters $a=24.4$ and $b=9.66 \; \mathcal{E}^{-1/2}$. The above approximation is compared to the full expression for $\mathcal{M}_\text{bh}$ in Fig \ref{fig:massfunction}. Using Eqs. \eqref{eq:massfunction} and \eqref{integral} an approximate expression for the optimal adiabatic ramp can be derived
\begin{equation}
\mathcal{V}_\bot (t) = \mathcal{V}_0 \ln^2 \left( \tfrac{t-t_0}{\tau_2} \right), \quad \forall t > \unit{1}{ms}.
\label{eq:analytic_bh}
\end{equation}
with $\tau_2 = \unit{651}{s}$, 
$\mathcal{V}_0/k_B=\unit{1.75}{nK}$, and $t_0=\unit{0.83}{ms}$. In \fig{fig:comparison} the closed form expressions for the optimal adiabatic ramp are compared to the numeric result showing excellent agreement in the respective time intervals. 

\begin{figure}
\centering
\includegraphics[scale=1]{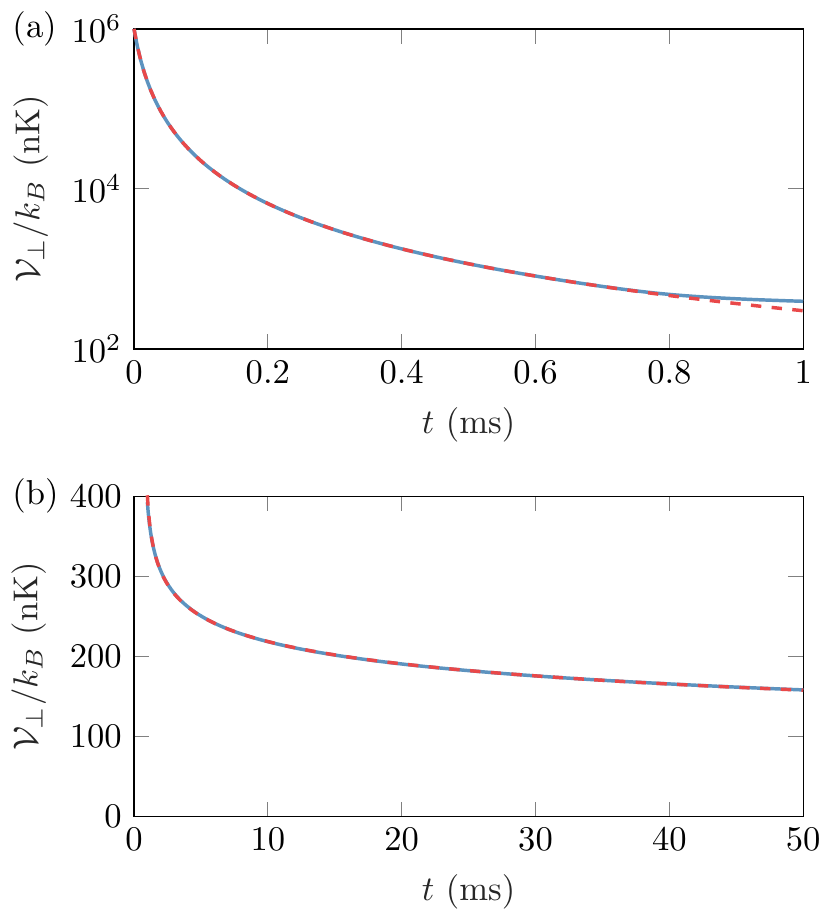}
\caption{Optimal adiabatic potential ramp $\mathcal{V}_\bot (t)$ versus time $t$: comparison between the numerical solution of the Euler-Lagrange equation (blue solid lines) and the analytic approximation (dashed red lines) given in Eqs. \eqref{eq:analytic_al} and \eqref{eq:analytic_bh}. For $t < \unit{0.7}{ms}$ (a) inter-band excitations dominate, whereas for $t> \unit{1}{ms}$ (b) intra-band excitations are the most relevant.}
\label{fig:comparison}
\end{figure}

We proceed by investigating the dependency of $\mathds{E}_\infty$ on the ramp duration $\tau$. It is worth noting that due to the structure of the Euler-Lagrange equation (cf. Eqs. \eqref{eq:adiasumquad} and \eqref{euler}) the optimal adiabatic ramp $\gamma'$ for a duration $\tau'$ can be obtained from a given optimal adiabatic ramp $\gamma$ for a duration $\tau$ via $\gamma' (t) = \gamma (t \tau/\tau')$ \footnote{We thank an anonymous referee for the interesting question about the scale invariance of the optimal ramp shape.}. In \fig{fig:fidelity} (a), $\mathds{E}_\infty (\tau)$ is shown for the bi-exponential and the optimal ramp shape. In both cases the data agrees very well with a $k \; \tau^{-2}$ dependency, with $k$ being a constant. Least square fits yield 
$k=\unit{8.01}{(ms)^2} $ and 
$k=\unit{1.66}{(ms)^2}$ for the bi-exponential and the optimal ramp respectively. 
This dependency can be explained 
by $\mathds{E}_\infty \propto (\partial_t \mathcal{V}_\bot)^2 \propto \tau^{-2}$, which also coincides with the result for the atomic limit given in \eq{eq:EinfAL}.

\begin{figure*}
\centering
\includegraphics[scale=1]{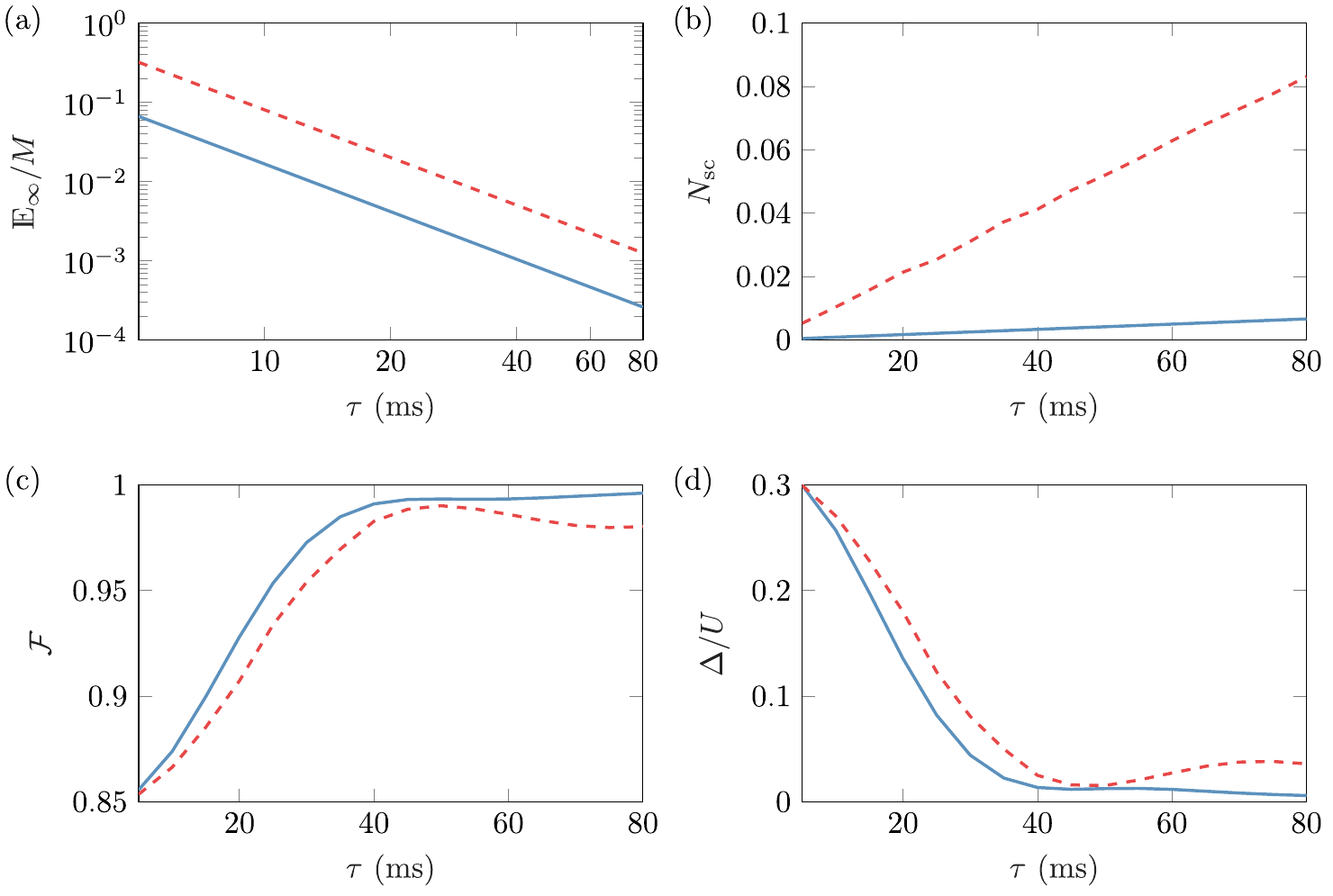}
\caption{Figures of merit for the bi-exponential (dashed red) and the optimal adiabatic (solid blue) ramp versus the ramp duration $\tau$: (a) maximal value $\mathds{E}_\infty/M$ of the adiabatic Lagrangian function per site during the ramp, (b) number of scattering events per atom $N_\text{sc}$ during the ramp, (c) transfer fidelity $\mathcal{F}$, (d) excess energy $\Delta$ of the final state in units of the interaction energy $U$ at $t=\tau$. The results shown in (c) and (d) are obtained from a many-body calculation for a 1D lattice with periodic boundary conditions, $N=8$ particles, and $M=8$ sites.}
\label{fig:fidelity}
\end{figure*}

\subsection{Impact of light scattering}

The physical process that limits the usage of long ramp durations is heating due to light scattering.
This effect has been studied in \cite{Gordon1980,Dalibard1985} and recently, with regard to optical lattices, in \cite{Pichler2010,Gerbier2010, Schachenmayer2014}. In order to estimate the impact of this process we calculate the number of scattering events per atom during the transfer process (cf. Appendix \ref{sec:heating}). \fig{fig:fidelity} (b) shows the dependency of $N_\text{sc}$ on the ramp duration $\tau$ for both the bi-exponential and the optimal ramp. The relation is linear with slopes of \unit{1.04}{s^{-1}} and \unit{0.08}{s^{-1}} for the bi-exponential and the optimal ramp respectively. Again, this can be explained using the time scale argument. The number of scattering events per atom can be reduced further by using light with a larger detuning, e.g. $\lambda_\bot = \unit{1064}{nm}$. However, already for the parameters used in this work an adiabatic transfer processes with negligible scattering can be realized.

\subsection{Fidelity of the transfer process}

In order to validate the adiabaticity of the transfer process we perform simulations of the many-body system using the calculated ramps. For this purpose we use the 1D single-band Bose-Hubbard model with periodic boundary conditions. This disregards possible excitations to higher bands. However, \fig{fig:adiabatictransfer} shows that these excitations are negligible for the majority of the ramp duration. 

We solve the time-dependent Schr\"odinger equation
\begin{equation}
\label{eq:mbse}
i \hbar \partial_t \ket{\psi (t)} = \hat{H}(t) \ket{\psi (t)},
\end{equation}
for the Hamilton operator $\hat{H}(J(t),U(t))$ given in \eq{eq:HJ}. The time dependence of the parameters $U$ and $J$ is determined by the ramp $\gamma (t) = \mathcal{V}_\bot (t)$ computed in Sec. \ref{sec:ramps}. In order to solve \eq{eq:mbse} we expand the system's state using the Fock basis (cf. Eqs. \eqref{eq:occupancy1}). This results in a system of ordinary differential equations
\begin{align}
i \hbar \dot{\psi}_{\vec{\eta}} (t) &= \sum\limits_{|\vec{\eta}'|=N} A_{\vec{\eta} \vec{\eta}'} (t) \; \psi_{\vec{\eta}'} (t), \\
A_{\vec{\eta} \vec{\eta}'} (t) &= \bra{\vec{\eta}} \hat{H} (t) + \hat{W} (t) \ket{\vec{\eta}'}.
\end{align}
The operator $\hat{W}(t)$ stems from the temporal change in the Wannier functions and is given by
\begin{align} 
\hat{W} (t) &= i \hbar \dot{\mathcal{V}}_\bot (t) \sum\limits_{n p} \sum\limits_{i j} C_{ij;1}^{np} (\mathcal{V}_\bot (t)) \; \hat{a}_{i}^{n\dagger} \hat{a}_{j}^p,
\end{align}
However, as stated earlier, this term does not induce intra-band excitation, \ie $C_{ij;1}^{nn}=0$ \cite{Pichler2013, Lacki2013a}. Therefore, we neglect it for our single-band simulation.

The initial state $\ket{\psi (0)}$ is the ground state of $\hat{H}(0)$. From the final state $\ket{\psi (\tau)}$ two figures of merit are obtained
\begin{align}
\mathcal{F} &= |\braket{\phi}{\psi (\tau)}|,\\
\Delta &= \bra{\psi (\tau)} \hat{H}(\tau) \ket{\psi (\tau)} - \bra{\phi} \hat{H} (\tau) \ket{\phi}.
\end{align}
Here, $\ket{\phi}$ is the ground state of the final Hamilton operator $\hat{H} (\tau)$, $\mathcal{F}$ is the transfer fidelity, and $\Delta$ is the energy difference between $\ket{\phi}$ and $\ket{\psi (\tau)}$. Figure \ref{fig:fidelity} (c) and (d) show the dependency of $\mathcal{F}$ and $\Delta $ on $\tau$ for the bi-exponential and the optimal adiabatic ramp. As expected, the transfer fidelity increases and the excess energy decreases for increasing ramp durations. This indicates a reduction of ramp-induced excitations. At $\tau \approx \unit{40}{ms}$ the slopes change significantly and saturation can be observed. In the case of the bi-exponential ramp this is accompanied by small amplitude oscillations indicating excitations due to non-adiabaticity. 

The calculations are performed with particle numbers up to $N=8$ and unit filling, \ie $M=N$. 
For ramp durations $\tau > \unit{40}{ms}$, both transfer fidelity and  excess energy are size independent. 

The results of this section show that a high transfer fidelity $\mathcal{F}>98\%$ can be achieved with ramp durations below \unit{50}{ms} and negligible photon scattering $N_\text{sc}<0.01$. It is worth noting that the ramp shape might be further improved by finding shortcuts to adiabaticity using optimal control \cite{Doria2011}. However, the presented approach has the advantage to result in simple and robust ramps.

\section{Limits on scalability}
\label{sec:scale}
\begin{figure}
\centering
\includegraphics[scale=1]{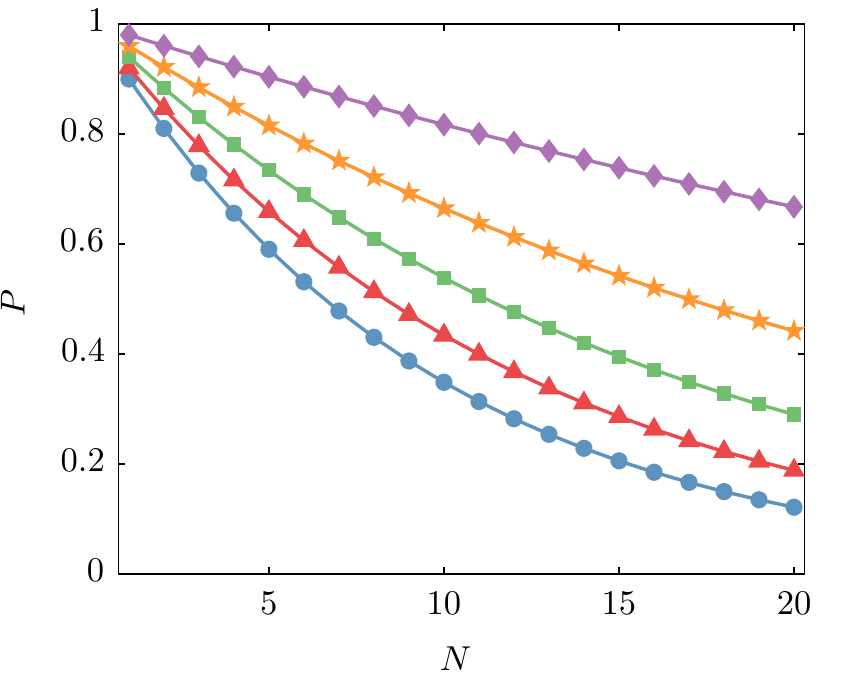}
\caption{Joint probability $P=p_0^N$ to prepare $N$ atoms in the motional ground state of $N$ isolated traps versus particle number N, for various single site success probabilities $p_0=0.9$ (blue circles), $p_0=0.92$ (red triangles), $p_0=0.94$ (green squares), $p_0=0.96$ (orange stars), and $p_0=0.98$ (violet diamonds).}
\label{fig:prepprob}
\end{figure}
Clearly, there are limitations for the maximum number of atoms that can be prepared. One limitation arises from the necessity to provide an array of many, sufficiently deep optical microtraps. With AODs, SLMs, or MLAs, and laser powers of a few watts, it is possible to produce arrays of a few hundred traps \cite{Barredo2016, Nogrette2014, Schlosser2011}. 

The next challenge is to prepare exactly one atom per trap. For arrays of up to $50$ microtraps, unit filling is experimentally feasible \cite{Endres2016, Barredo2016}. According to Ref.~\cite{Barredo2016}, this could be extended to a few hundred traps using state-of-the-art technology. 

Another prerequisite of the discussed scheme is the preparation of atoms in the motional ground-state with high fidelity. Using Raman sideband cooling, an  occupation probability of $p_0=90 \%$ has been achieved \cite{Kaufman2012}. This value was limited by a weak confinement in axial direction. Application of additional axial confinement, as considered in this work, should enhance the probability. However, if the technique is applied in parallel to an $N$-trap array then the  joint success probability $P=p_0^N$ to cool all atoms to the motional ground state decreases exponentially. This trend is shown in \fig{fig:prepprob} for several values of $p_0$ and constitutes the biggest challenge on the path to large atom numbers.

\section{Conclusions and outlook}
\label{sec:conclusions}

In conclusion, we have analyzed a preparation scheme for a Mott insulator state in the itinerant regime, starting from an ensemble of individual atoms in the atomic limit. 
On this behalf, the depth of the optical potential is ramped down significantly. 
In order to minimize both ramp-induced excitations and external heating during this process, we propose a variational procedure to obtain optimal rapid adiabatic ramp shapes.  
The choice of error functionals is physically motivated by the adiabatic theorem and can be generalized readily to optimize 
multi-dimensional time-dependent control parameters. 
In comparison to a full optimal control procedure \cite{Doria2011},
the presented approach is simple and robust.

For realistic experimental parameters, we investigate the fidelity 
of the resulting optimal ramps and asses the detrimental impact of 
spontaneous photon scattering. 
This demonstrates the feasibility of the proposed scheme with state-of-the-art technology. These conclusions are based on simulations of the one-dimensional Bose-Hubbard model. However, we expect similar results for two dimensions taking into account the scaling of the adiabatic error function \eq{eq:Mbh} with the coordination number.

If the depth of the microtrap array is reduced beyond the point discussed in this work, then first the superfluid phase of the Bose-Hubbard model and finally a BEC can be prepared. Here, the analysis of the preparation process based on the adiabatic theorem breaks down, because the energy gap between the ground state and the lowest excited state vanishes. However, this process corresponds to the time-reversed loading scheme used in optical lattices. The feasibility of this approach for microtrap arrays with similar parameters as discussed in this work is shown in \cite{Sturm2017}. This opens an alternative route for the preparation of BECs by direct laser cooling \cite{Hu2017}, which is especially appealing for the investigation of atomic and molecular species that can not be cooled evaporatively.

\section*{Acknowledgments}
M.R.S. and R.W. acknowledge support from the German Aeronautics and Space Administration (DLR) through Grant No. 50 WM 1557. M.S. and G.B. acknowledge support by the Deutsche Forschungsgemeinschaft (DFG) through Grant No. BI 647/6-1 within the Priority Program SPP 1929 (GiRyd).

\appendix

\section{Equivalence of minimizing $\mathds{E}_1$ and $\mathds{E}_{\infty}$}
\label{sec:proof}

In this appendix we consider the situation of \sect{sec:mott}. In the time interval $[0,\tau]$ the system is controlled by a one-dimensional and monotonically decreasing parameter curve $\gamma (t)$ with $\gamma (0) = \gamma_i$ and $\gamma (\tau) = \gamma_f$. The Lagrange function
\begin{equation}
	\mathcal{L}(\gamma,\dot{\gamma}) = \frac{1}{2} \mathcal{M}(\gamma) \dot{\gamma}^2
\end{equation}
is convex and the mass function $\mathcal{M}(\gamma)$ is sufficiently smooth, monotonically decreasing, and positive. Under these conditions the following proposition can be stated.

\theoremstyle{definition}
\newtheorem*{prop}{Proposition}
\begin{prop}
The parameter curve $\gamma_0 (t)$ is a minimum of the functionals $\mathds{E}_1 [\gamma,\dot{\gamma}]$ and $\mathds{E}_\infty [\gamma,\dot{\gamma}]$, if $\gamma_0 (t)$ satisfies the corresponding Euler-Lagrange equation \eqref{euler}.
\end{prop}

\begin{proof}
The fact that $\gamma_0 (t)$ is also a minimum of $\mathds{E}_1 [\gamma,\dot{\gamma}]$ follows from the convexity of the Lagrange function \cite{Lopez2012}. In order to prove that $\gamma_0 (t)$ is a minimum of $\mathds{E}_\infty [\gamma,\dot{\gamma}]$ the following intermediate steps are used.
\begin{itemize}
\item[\textbf{I}] $\mathcal{L}(\gamma_0(t),\dot{\gamma}_0(t))$ is constant for $t\in [0,\tau]$.
\item[\textbf{II}] $\gamma_0(t)$ minimizes the functional $\max_t \sqrt{\mathcal{L}(\gamma,\dot{\gamma})}$.
\end{itemize}
Using the fact that $\gamma_0$ satisfies the Euler-Lagrange equation, it is straight forward to show \textbf{I},
\begin{align}
\frac{d}{d t} \mathcal{L} (\gamma_0 (t), \dot{\gamma}_0 (t)) 
&= \mathcal{M}(\gamma_0) \dot{\gamma}_0 \ddot{\gamma}_0 + \frac{\dot{\gamma}_0^3}{2} \frac{\partial \mathcal{M}}{\partial \gamma}\bigg|_{\gamma=\gamma_0}\\
&= \dot{\gamma}_0 \left( \frac{d}{d t} \frac{\partial \mathcal{L}}{\partial \dot{\gamma}} - \frac{\partial \mathcal{L}}{\partial \gamma} \right)\bigg|_{\gamma=\gamma_0} = 0.
\end{align}
The key to prove the implication $\textbf{I} \Rightarrow \textbf{II}$ is to observe that $\int_0^\tau \sqrt{\mathcal{L}} \, \dd t$ is a geometric invariant and therefore path-independent
\begin{align}
\int\limits_0^\tau \sqrt{\mathcal{L}(\gamma(t),\dot{\gamma}(t))} \; d t &= - \int\limits_0^\tau \dot{\gamma}(t) \sqrt{\tfrac{1}{2} \mathcal{M}(\gamma(t))} \; d t\\
&= -\int\limits_{\gamma_i}^{\gamma_f} \sqrt{\tfrac{1}{2} \mathcal{M}(\gamma)} \; d \gamma.
\end{align}
From this observation follows that the functional $\max_t \sqrt{\mathcal{L}(\gamma,\dot{\gamma})}$ is minimized if $\sqrt{\mathcal{L}}$ is constant. The latter is obviously true if $\mathcal{L}$ is constant, which is achieved by the parameter curve $\gamma_0$ (cf. statement \textbf{I}). Consequently, $\gamma_0$ minimizes the functional $\max_t \sqrt{\mathcal{L}(\gamma,\dot{\gamma})}$, \ie $\textbf{I} \Rightarrow \textbf{II}$. It is apparent that $\mathbf{II}$ implies that $\gamma_0 $ is a minimum of $\mathds{E}_\infty [\gamma, \dot{\gamma}]$ since the function $x \mapsto x^2$ is monotonically increasing for $x>0$.
\end{proof}

\section{Variational ground-state in a two-dimensional Gaussian potential}
\label{sec:variationcalc}

We consider a Gaussian variational ansatz of the form
\begin{equation}
	\varphi (x,y) = \frac{1}{\sqrt{\pi} a_\bot} \exp \left( -\frac{x^2+y^2}{2 a_\bot^2} \right)
\end{equation}
for the ground state of a two-dimensional Gaussian potential well
\begin{equation}
	V(x,y)=\mathcal{V}_\bot \exp \left(-2 \frac{x^2+y^2}{w^2_{0\bot}} \right).
\end{equation}
We determine the wave function's width $a_\bot$ by minimizing the energy functional
\begin{align}
	E(a_\bot) &= \iint \varphi (x,y) \; H_{2D} \; \varphi (x,y) \; \dd x \; \dd y\\
	&= \frac{\hbar^2}{2 m a_\bot^2} - \mathcal{V}_\bot \frac{w_{0\bot}^2}{w_{0\bot}^2 + 2 a_\bot^2}.
\end{align}
This yields the expression given in \eq{eq:abot}. In the above equation the position representation of the two-dimensional single particle Hamilton operator is used
\begin{equation}
	H_{2D} = -\frac{\hbar^2}{2m} (\partial_x^2 + \partial_y^2) + V(x,y).
\end{equation}
A similar calculation for a one-dimensional Gaussian well is given in \cite{Nandi2010}.

\section{Scattering rates and light shifts }
\label{sec:heating}
In the time interval $[0,\tau]$, the cumulative number of scattered photons per atom is given by
\begin{equation}
N_\text{sc}(\tau) = \int_0^\tau \Gamma_\text{sc} (t) \; \dd t.
\end{equation}
For alkali atoms, trapped by far-off-resonant, linearly polarized laser beams with intensity $I$ and angular frequency $\omega$, the scattering rate $\Gamma_\text{sc}$ and the light shift $V$ read \cite{Grimm2000}
\begin{align}
\begin{split}
\Gamma_\text{sc}  =&  
\frac{\pi c^2 I \omega^3}{2 \hbar} \left[ \frac{\Gamma_1
^2}{\omega_1^6} \left( \frac{1}{\omega -  \omega_1} - \frac{1}{\omega +  
\omega_1}  \right)^2 \right.  \\
 &+\left. \frac{2\Gamma_2^2}{\omega_2^6} \left( \frac{1}{\omega -  \omega_2} -
 \frac{1}{\omega +  \omega_2}  \right)^2 \right],
\end{split}\\
\begin{split}
V =&  \frac{\pi c^2 I}{2} \left[ \frac{\Gamma_1}{\omega_1^3} \left( \frac{1}{
\omega -  \omega_1} - \frac{1}{\omega +  \omega_1}  \right) \right. \\
&+\left. \frac{2\Gamma_2}{\omega_2^3} \left( \frac{1}{\omega -  \omega_2} - 
\frac{1}{\omega +  \omega_2}  \right) \right].
\end{split}
\end{align}
The decay rates $\Gamma_i$ and the transition frequencies 
$\omega_i$ correspond to the D1 and D2 lines of the respective species. 


%

\end{document}